\documentclass{cccg26} 
\usepackage{amssymb,amsmath}
\usepackage{nicefrac}
\usepackage{color}
\usepackage[american]{babel}
\usepackage{graphicx}
\usepackage{version}
\usepackage{subcaption}
\usepackage[disable]{todonotes}
\usepackage{url}
\usepackage{booktabs} % nice tables
\usepackage{enumitem} % configurable enumerate
\usepackage{hyperref} 
\usepackage{thm-restate}

\newif\iffull
\fulltrue
\graphicspath{{figures/}}

\newtheorem{claim}{Claim}

\newcommand{\calT}{{\ensuremath{\cal T}}}
\newcommand{\calS}{{\ensuremath{\cal S}}}

\renewcommand{\d}{\ensuremath{s}}

\newcommand{\leaveout}[1]{}

\date{}
\title{Unbent collections of non-planar $\d$-grid-drawings%\thanks{Many thanks to the authors of \cite{AnticEtAlWG25} for inspiring conversations.}
}
\author{Therese Biedl 
\footnote{
David R.~Cheriton School of Computer
Science, University of Waterloo, Waterloo, Ontario N2L 3G1, Canada.  \textit{biedl@uwaterloo.ca}.
Supported by NSERC.  Many thanks to the authors of \cite{AnticEtAlWG25}, and especially Todor Anti\'{c}, Giuseppe Liotta and Sascha Wolff, for helpful input.}
}
\begin{document}

\maketitle

\begin{abstract}
In a recent paper, Anti\'{c} et al.~studied 
collections of planar orthogonal drawings of a graph where every edge is unbent in at least one drawing.   This paper generalizes this concept to non-planar drawings, and shows that then two drawings always suffice (for planar drawings three drawings are sometimes needed).   The results can also be generalized to $s$-grid drawings for $s\geq 3$.
\end{abstract}

\section{Introduction}

In the standard \emph{graph drawing} problem, we are given a graph and the goal is to create a 
drawing %that helps to understand the graph    
(often a node-link-diagram in $\mathbb{R}^2$, but other graph drawing styles and/or higher dimensions have been studied). 
%The literature in this field is extensive and cannot all be reviewed here, 
See various textbooks and handbooks \cite{DBETT98full,KW01,NR04,Tam13} for an overview of the field.

If the given graph is very complex, then creating one legible drawing is nearly impossible.   For this reason, Hlin\v{e}n\'{y}  and Masa\v{r}\'{i}k \cite{HlinenyM23} proposed a different paradigm, where we are given a graph $G$ and we want to create a \emph{collection of drawings} $\Gamma_1,\dots,\Gamma_k$ of $G$ for some $k\geq 1$ that all together satisfy some property that is useful for visualization purposes.   Hlin\v{e}n\'{y} and Masa\v{r}\'{i}k specifically studied collections $\Gamma_1,\dots,\Gamma_k$ of node-link diagrams under the constraint that every edge is drawn without crossing in at least one of $\Gamma_1,\dots,\Gamma_k$. They showed that it is NP-hard to minimize the \emph{order} $k$, and gave various asymptotic bounds on the order that can be achieved for complete and other dense graphs.

The idea of collections of drawings was picked up by Anti\'{c}, Liotta, Masa\v{r}\'{i}c, Ortali, Pfretzschner, Stumpf, Wolff and Zink \cite{AnticEtAlWG25} who studied a very different drawing style.   Specifically, they considered \emph{planar graphs} (graphs that can be drawn in the plane without crossing) and \emph{orthogonal drawings} (a drawing style that represents vertices as points and edges as polygonal curves that use only horizontal and vertical segments); furthermore the drawings must be planar, i.e., not have any crossings.    Such drawings can exist only if the graph is a \emph{4-graph}, i.e., every vertex has at most four incident edges.      Orthogonal drawing have a long history, see for example an overview article by Duncan  and Goodrich and the references therein \cite{DuncanG13.crossref}.

The goal in the paper of Anti\'{c} et al.~\cite{AnticEtAlWG25} was to create an \emph{unbent collection} of planar orthogonal drawings, i.e., they want to find planar orthogonal drawings $\Gamma_1,\dots,\Gamma_k$ such that every edge $e$ is drawn \emph{without bend}
in at least one of the drawings. They argued that for some planar 4-graphs (including the octahedron, the graph of Figure~\ref{fig:octahedron}), this cannot be achieved with at most two planar drawings. On the other hand it is easy to show that $k=3$ drawings always suffice.   
%(They also study the total number of bends in the drawings, a topic that we will not consider here.)

\begin{figure}[ht]
\hspace*{\fill}
\includegraphics[page=5,scale=0.8]{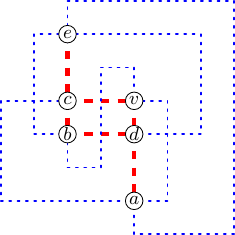}
\hspace*{\fill}
\includegraphics[page=7,scale=0.8]{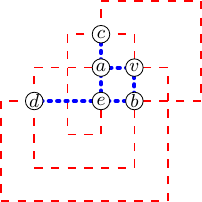}
\hspace*{\fill}
\caption{Two orthogonal drawings of the octahedron; every edge is drawn without bend (thick)
in one of the drawings.  No two \emph{planar} drawings of the octahedron with this property
exist \cite{AnticEtAlWG25}.}
\label{fig:octahedron}
\end{figure}

Inspired by the paper by Antic et al.~\cite{AnticEtAlWG25}, we ask a very similar question, but for graphs that are not necessarily planar (and so the drawings also need not be planar).   Thus, given a 4-graph $G$, we want to find an unbent collection of (not necessarily planar) orthogonal drawings of $G$ 
whose order is small.
It is very easy to see that order $3$ is always possible.    But as the main result of this paper, we show that in fact order $2$ suffices.

\begin{theorem}
\label{thm:un_lt_2}
\label{thm:orth}
Every 4-graph $G$ has an unbent collection of orthogonal drawings of order  2.   
%Put differently, $G$ has two orthogonal (not necessarily planar) drawings $\Gamma_1,\Gamma_2$  
%such that every edge of $G$ is drawn without bend in at least one of them.   
%Such drawings can be found in $O(n)$ time.
%\future{can we actually find the split in linear time?   Think about this as you re-do the proof.   If yes, brag about it here.}
\end{theorem}

Note that `two' is obviously best-possible; there are many 4-graphs that do not have an orthogonal drawing without bends (consider
a triangle, for example, or graph $K_{2,3}$).    Testing whether order one suffices 
%Naturally one wonders whether we could perhaps even achieve this with one drawing, which then would have to be a \emph{rectilinear drawing} (where all edges are
for a specific graph $G$ is NP-hard, since this is the same as testing
whether $G$ has an orthogonal drawing without bends, which was shown to be NP-hard  by
Eades, Hong, and Poon~\cite{EadesHP09}. 
% TB: Mostly for myself.   Reduction from 3-sat.   Build a frame that fixes that variables are left while clauses go right.   Variable-gadgets are also quite fixed, but can stretch upwards and flip the entire thing left-right.   Clause-gadgets are just squares with connections to the variables that can have some bends, and since we can change the embedding, these squares have a lot of flexibility.   It works out such that two of the connections to variable-gadgets can always be realized, but for the third one the correct literal-side of the variable-gadget must have been turned towards the clauses, i.e., the literal must be true. 
%They also have some FPT results, but the parameter is close to $n$, so not a very helpful result.  
%We could try to show that the problem is FPT w.r.t.\ vertex cover number (vcn), but again vcn is linear in $n$ since the maximum degree 4.  
%Is the planar version FPT in pathwidth?  
%(If the parameter is the height~$h$ of the drawing, then we can minimize the width of a given orthogonal representation by dynamic programming in $n^{O(h)}$ time~\cite{dgklwz-paoc-JCSS26}.)   

A natural generalizations of orthogonal drawings are \emph{$\d$-grid drawings}.   Formal definitions will be given below, but the most prominent examples are \emph{hexagonal} ($\d=3$) and \emph{octilinear drawings} ($\d=4$), see also Figure~\ref{fig:K7}.   Such drawings can exist only for $2\d$-graphs, i.e., graphs with maximum degree at most $2\d$.
Much less is known about these types of drawings, but see for example 
\cite{AB04full,BBB+08,Kant93-hexagonal,Tol90} 
and \cite{BekosFK19,BekosGKK15, BekosKK17, Nollenburg05}  for results on hexagonal and octilinear drawing results, respectively,
and \cite{KeszeghPP13} for planar drawings with $\d$ slopes. 

\begin{figure}[ht]
\hspace*{\fill}
\includegraphics[page=1,scale=0.8,angle=180]{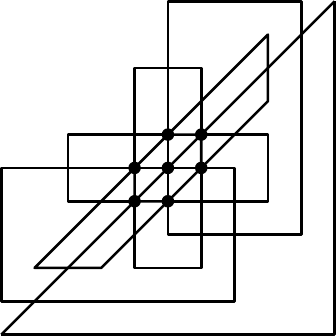}
\hspace*{\fill}
\includegraphics[page=1,scale=0.8,angle=180]{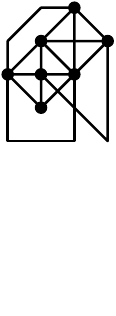}
\hspace*{\fill}
\caption{A hexagonal drawing of $K_7$, and a small octilinear drawing.} 
\label{fig:K7}
\end{figure}

There appear to be no previous results on unbent collections of $\d$-grid drawings.
We show that the results for orthogonal drawings transfer to $\d$-grid drawings. 
%both our Theorem~\ref{thm:un_lt_2} as well as the results of Anti\'{c} et al.
%So generalizing from orthogonal drawings to grids with slopes of gridlines, 
%we study \emph{unbent collections of $\d$-grid drawings}, i.e., collections $\Gamma_1,\dots,\Gamma_k$ of (not necessarily planar) $\d$-grid drawings of a $2\d$-graph $G=(V,E)$. The aim is to minimize the size $k$ of the collection.   
%It is very easy to see that these always exist, and order $d{+1}$ always suffices,
%because a $2\d$-graph $G$ can have at most $dn$ edges, and hence has an edge-partition into $d{+}1$ trees $T_1,\dots,T_{d+1}$.   Each tree $T_i$ can be drawn without bend, the edges of $E\setminus T_i$ can be added into this drawing with sufficiently many bends (see also Section~\ref{sec:insert_edge}, and the result hence follows.   Theorem~\ref{thm:un_lt_2} hence states that we can \emph{beat} this upper bound for $d=2$, since we only need $d$ drawings.   This turns out to be true for the other drawing styles as well (and is in fact much easier to prove).
%\future{We can do collections of $\d$-grid drawings of order $d+1$ always.   That's worth saying and gives an excuse to prove the thingy of Antic et al easier.   Write!}

\begin{theorem}
\label{thm:hex_oct_lt_2}
\label{thm:hex_oct}
For $\d\geq 3$, every $2\d$-graph has an unbent collection of $\d$-grid drawings of order $s$.
%Put differently, there are $\d$-grid drawings $\Gamma_1,\dots,\Gamma_d$ of $G$
%such that every edge of $G$ is drawn without bend in at least one of them.   
%Furthermore, in the drawings no three edges cross in a point and no two bends coincide.
%Such drawings can be found in $O(n)$ time.
\end{theorem}

\begin{theorem}
\label{thm:hex_oct_planar}
For $\d\geq 3$, every planar $2\d$-graph has an unbent collection of planar $\d$-grid drawings of order $s{+}1$.
%Put differently, there are $d$-grid drawings $\Gamma_1,\dots,\Gamma_d$ of $G$
%such that every edge of $G$ is drawn without bend in at least one of them.   
%Furthermore, in the drawings no three edges cross in a point and no two bends coincide.
\end{theorem}

Clearly again order two or more is needed for the unbent collections of $\d$-grid drawings: $K_{2,\d{+}1}$
does not have an $\d$-grid drawing without bend.   (On the other hand, a triangle can be drawn without
bend for $\d\geq 3$, which is the main reason why the proof for $\d\geq 3$ is simpler than for $\d=2$.)
Testing whether a graph has an $\d$-grid drawing without bends is NP-hard (see e.g.~\cite[Thm.~5.3]{Nollenburg05} for hardness of
having one plane octilinear drawing; this can be generalized to other types of grids).
%We do leave a gap between the order that we can always achieve and the order that is sometimes known to be necessary.

\section{Preliminaries}

We assume familiarity with graphs, see e.g.~\cite{Die12full}.   
In this paper, the input graph $G$ always has $n$ vertices and is assumed to be simple.
A \emph{drawing} of $G$ is an assignment of distinct
points to vertices of $G$, and a polygonal curve to every edge such that the curve of $e=(v,w)$ connects
the two points of $v$ and $w$, contains no other vertex-points, and any two edge-curves share a 
finite number of points.    
%We consider the edge-curve to be open at its ends, i.e., to not include the points of $v,w$.    
Where no confusion arises, 
the graph-theoretic name (`vertex', `edge') also refers to its geometric representation (`point', `curve').
%An \emph{edge-segment} (of an edge $e$) is a maximal straight-line segment within the 
%curve of $e$. The ends of an edge-segment of $e$ are either an endpoint of $e$ or a 
The place where an edge-curve changes direction is called a \emph{bend}, and an edge is called
\emph{unbent} if it has no bends. A drawing is called \emph{unbent} if all its edges are.
An \emph{unbent collection} (of drawings of some type that will be clear from context) is a set
of such drawings of the same graph such that every edge is unbent in at least one of them.

A \emph{planar drawing} of a graph is a drawing such that edge-curves do share not points except
at their endpoints;   the graph that has such a drawing is called \emph{planar}.
A \emph{plane graph} is a planar graph with a fixed \emph{rotation scheme}, i.e., for every vertex
$v$ there is a fixed cyclical order $\rho(v)$ of the incident edges.   A \emph{plane drawing}
of such a graph  is a drawing where the incident edges of every vertex $v$ appear in the order of 
$\rho(v)$ when enumerated in the ccw order in which they leave $v$.

Let $\calS\subseteq \mathbb{Q}\cup \{\infty\}$ be a finite set of \emph{slopes}.   
The \emph{gridlines} of an \emph{$\calS$-grid} are all those lines $\ell$ that have a slope in $\calS$ 
and that pass through a \emph{grid point}, i.e., a point with integral coordinates.     
(Since slopes are rational or infinity, any gridline passes through infinitely many grid points.)
The \emph{hexagonal grid} uses $\calS=\{0,1,\infty\}$%
\footnote{This is, up to a shear, equivalent to the more commonly used approach of defining the hexagonal
grid to have lines of slope $\{0,\pm \sqrt{3}/2\}$.}
 while the \emph{octilinear grid} uses $\calS=\{0,\pm 1,\infty\}$.
%For ease of description, the term
An \emph{$\d$-grid} (for some integer $\d\geq 2$) refers to an $\calS$-grid for some set $\calS$ with $|\calS|=\d$;
the specific set $\calS$ will not matter for the arguments.
%(some of the later details will depend on the specific set $\calS$ and not only its cardinality).

A drawing of a graph $G$ is an \emph{$\d$-grid drawing} 
if every vertex-point and bend is at a grid point
and every edge-segment lies on a gridline of the $\d$-grid.     For every 
%vertex $v$ (and every grid point more generally) 
grid point
there are $2\d$ incident \emph{gridrays}, i.e., rays that go along gridlines,
and edges can only leave $v$ along gridrays.   Therefore $\d$-grid drawings can only exist
if every vertex has degree at most $2\d$; graphs that satisfy this are called \emph{$2\d$-graphs}.
%The \emph{bounding box} of a $\d$-grid drawing is the smallest axis-aligned box that has all points
%of vertices and bends inside; the number columns and rows intersecting the bounding box is also
%called the \emph{width} and \emph{height} of a drawing, and the maximum of width and height is
%called the \emph{size}.

%Consider a vertex $v$ in a $\d$-grid drawing, and the $d$ gridlines through $v$.
%A \emph{free gridray} is a ray with origin $v$ and along a gridline such that no edge
%occupies the part of it immediately adjacent to $v$.   (Later parts of the ray may well contain
%edge curves and/or vertex points.)     If $v$ has degree $k<2\d$ then it has $2\d-k$ free gridrays.

A few graph-theoretic definitions are needed.  A \emph{pseudotree} is a connected
graph that is a tree after deleting at most one edge.   
%A \emph{forest} is a graph that has no cycles; it is called a \emph{tree} if it
%is furthermore connected.   
A \emph{pseudoforest} is a graph $G$ where every connected component is a pseudotree. 
A \emph{cactus} is a graph for which all cycles
are edge-disjoint.  A \emph{strict cactus} is a graph for which all cycles are vertex-disjoint.   
Every pseudoforest is a strict cactus.
Use the prefix `$2\d$-' for such graphs if their degrees are at most $2\d$.

An \emph{edge-partition} of a graph $G=(V,E)$ is a partition of $E$ with every
edge appearing in exactly one of the sets.   The \emph{arboricity} [\emph{pseudoarboricity}] of
a graph is the smallest number $k$ of sets in an edge-partition $E_1\cup \dots \cup E_k$ such that every subgraph $(V,E_i)$
is a forest [pseudoforest].     The following is known:

\begin{theorem}[\cite{NW61,PicardQ82}]
\label{thm:NW}
%\todo[inline]{Read Picard/Queyranne at UW}
%Let $G$ be a graph with at least one edge.
The arboricity of a graph $G$ is 
$\max_H \lceil \tfrac{|E(H)|}{|V(H)|-1}\rceil $
while its pseudoarboricity is $\max_H \lceil \tfrac{|E(H)|}{|V(H)|}\rceil$,
where the maximum is taken over all subgraphs $H$ of $G$ 
with enough vertices to make the denominator non-zero.
%with at least two [one] vertices, respectively.
%
%\begin{center}
%\vspace*{-3mm}
%{$\lceil \max\{ \tfrac{|E(H)|}{|V(H)|-1}: H\text{ is a subgraph of $G$ with $|V(H)|\geq 2$}\}\rceil$} 
%\vspace*{-3mm}
%\end{center}
%
%\noindent{}while its pseudoarboricity is  
%
%\begin{center}
%\vspace*{-3mm}
%{$\lceil \max\{ \tfrac{|E(H)|}{|V(H)|}: H\text{ is a non-empty subgraph of $G$}\}\rceil.$}
%\end{center}
\end{theorem}

\section{Proof-skeleton}
\label{sec:skeleton}

This section gives the proof of the main result: For $\d\geq 2$
%(Theorems~\ref{thm:orth} and~\ref{thm:hex_oct}):  
any $2\d$-graph $G$ has a small unbent collection of $\d$-grid drawings.
Some of the steps in these proofs
are non-trivial (and possibly of interest in their own right) and their proofs
are hence delayed till later.

As a first step, convert the
problem of finding unbent collections into an edge-partition problem.
%, and show the following:

\begin{lemma}
\label{lem:partition_to_drawing}
Let $G=(V,E)$ be a $2\d$-graph with an edge-partition $E=E_1\cup \dots \cup E_k$
such that 
	$G_i=(V,E_i)$ has an unbent $\d$-grid-drawing $\Gamma_i'$ 
	for $i\in [k]$.. 
Then $G$ has an unbent collection $\Gamma_1,\dots,\Gamma_k$ of $\d$-grid-drawings. 
\end{lemma}
\begin{proof}
View each $\Gamma_i'$ as a drawing of $G$ where some edges (namely, those in $E\setminus E_i$)
are \emph{missing}.    So it remains to add the missing edges into drawing $\Gamma'_i$.
It is folklore how to do this for orthogonal drawings by adding at most two rows and two columns
per missing edge, see for example related discussions in \cite{PapakostasT95}.  
For $\d$-grid drawings, no previous discussions of how to add missing edges seem to exist, 
	but it is not hard to see that this can be done after scaling the drawing suitably;
details are in Section~\ref{sec:insert_edges}. 
\end{proof}

Note here Lemma~\ref{lem:partition_to_drawing} is \emph{not} true if 
the resulting collection must consist of \emph{planar} drawings;  this is the
main reason that Theorem~\ref{thm:hex_oct_planar} requires one more drawing
than Theorem~\ref{thm:hex_oct}.   
%Details of how to handle planar drawings are are given in Section~\ref{sec:insert_edges}. 

With Lemma~\ref{lem:partition_to_drawing} in hand, the next step is to find
subgraphs of $G$ that have unbent $\d$-grid-drawings. 
The following result is well known for trees and orthogonal drawings, and can easily be transferred
to strict cacti and $\d$-grid drawings. See Section~\ref{sec:easy_to_draw} for details
and the thicker lines in Figure~\ref{fig:octahedron} for an example for triangle-free pseudoforests.
%For example, the thicker lines in Figure~\ref{fig:octahedron} show two planar unbent orthogonal
%drawings of two triangle-free pseudoforests.

\begin{restatable}{lemma}{EasyToDraw}
\label{lem:easy_to_draw}
For all $\d$, every plane $2\d$-forest has a plane unbent $\d$-grid drawing.  \\
If $\d\geq 3$, then every strict $2\d$-cactus has a planar unbent $\d$-grid-drawing.  \\
Every triangle-free strict 4-cactus has a planar unbent orthogonal drawing.  
\end{restatable}

The final piece of the puzzle is to find an edge-partition into suitable subgraphs.   This is very easy if 
subgraphs need not be triangle-free.

\begin{lemma}
\label{lem:pseudoforest}
Every $2\d$-graph $G$ has an edge-partition into $\d$ pseudoforests, and
an edge-partition into $\d+1$ forests.
\end{lemma}
\begin{proof}
Every subgraph $H$ of $G$ has maximum degree $2\d$, and hence $|E(H)|\leq \d h$
where $h=|V(H)|$.  So by Theorem~\ref{thm:NW}  the pseudoarboricity of $G$ is at most $\d$.

As for the arboricity, if $h\geq \d+1$ then $\d h\leq (\d{+}1)(h{-}1)$.
If $h\leq \d$ then any subgraph $H$ of $G$ with $h$ vertices has at most $h(h-1)/2\leq \d(h-1)/2\leq (\d+1)(h-1)$ edges by simplicity.
So the arboricity of $G$ is at most $\d{+}1$ by Theorem~\ref{thm:NW}.
\end{proof}

With this the proof of Theorem~\ref{thm:hex_oct} is finished: Split the edges of the given $2\d$-graph $G$ into $\d$ pseudoforests with Lemma~\ref{lem:pseudoforest}, each of them has an unbent $\d$-grid drawing by Lemma~\ref{lem:easy_to_draw}, and by Lemma~\ref{lem:partition_to_drawing} the theorem holds.   
The proof of Theorem~\ref{thm:hex_oct_planar} is very similar, but uses an edge-partition into forests and needs some extra care to preserve planarity, see Section~\ref{sec:insert_edges}.

For $d=2$ (i.e., Theorem~\ref{thm:orth}) the proof is a bit harder, because here it does not suffice to take any partition into pseudoforests; we have to find \emph{triangle-free} subgraphs.   It is unknown whether a partition into triangle-free pseudoforests always exists, but the following weaker result holds (the proof is non-trivial and in Section~\ref{sec:4graph_cactus}).

\begin{lemma}
\label{lem:4graph_cactus}
Every $4$-graph has an edge-partition into two triangle-free strict cacti.
\end{lemma}

Combining this with Lemma~\ref{lem:partition_to_drawing} and~\ref{lem:easy_to_draw} 
as before then proves Theorem~\ref{thm:orth}.

\section{Splitting a 4-graph into triangle-free strict cacti}
\label{sec:4graph_cactus}

This section gives the proof of Lemma~\ref{lem:4graph_cactus}.   So fix a 4-graph $G$
and proceed by induction on the number $n$ of vertices.   If $n\leq 5$, then by
simplicity $G$ is a subgraph of $K_5$.   One can easily find an edge-partition of $K_5$
into two 5-cycles, and the corresponding edge-partition of $G$ therefore consists
of triangle-free pseudoforests, which proves the result. 
Now assume that $n\geq 6$ and consider the following cases.

\medskip\noindent{\bf Case 1:} $G$ has an edge cut $E'$ of size at most two,
i.e., $G':=(V,E\setminus E')$ has multiple connected components.    
Write  $G'$ as the disjoint
union of two non-empty subgraphs $H^{1}$ and $H^2$ such that the edges between them are exactly $E'$, see Figure~\ref{fig:edgeCut}.
By induction, for $j\in [2]$ the graph~$H^j$ has an edge-partition $E^j_1\cup E^j_2$ into triangle-free strict cacti.
Since $H^1$ and $H^2$ are disjoint and disconnected, we can combine their
edge-partitions (i.e., define $E_i':=E_i^1\cup E_i^2$ for $i\in [2]$) to get an
edge-partition $E_1'\cup E_2'$ of $G'$ into two triangle-free strict cacti.    

If $E'=\emptyset$, then we are done, so 
assume that $E'$ is non-empty and add one of its edges (say $e_1$)
to $E_1'$ and the other edge $e_2$ (if if exists) to $E_2'$.    Let $G_i$ (for $i\in [2]$) be the
graphs formed by the resulting edge sets.    Observe that $e_1$ is a bridge in $G_1$,
because it is the only edge in $E_1'\cup \{e_1\}$ that connects $H^1$ to $H^2$.   Therefore $e_1$
belongs to no cycle,   or in other words, $G_1$ has the same set of cycles as 
%the graph formed by $E_1'$ and 
$(V,E_1')$ and is hence a triangle-free strict cactus.
\footnote{Note that $G_1$ is not necessarily a pseudoforest, even if $(V,E_1')$ was,
and it is not clear whether a stronger induction hypothesis (and more careful
assembly of the sets $E^j_{\ell}$ and $E'$ into $G_1$ and $G_2$) could lead to
a partition into triangle-free pseudoforests.}
Similarly $G_2$ is a triangle-free strict cactus.

\begin{figure}[ht]
%\hspace*{\fill}
%\includegraphics[scale=0.6,page=4]{problem.pdf}
\hspace*{\fill}
\includegraphics[scale=0.8,page=6]{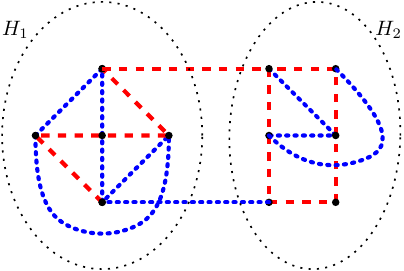}
\hspace*{\fill}
\caption{The case where $G$ has a small edge-cut. 
}
\label{fig:smallCut}
\label{fig:edgeCut}
\end{figure}

\bigskip

For all following cases, it is assumed that Case 1 does \emph{not} apply.
In particular therefore
$G$ is connected, and in fact, it is 3-edge connected  and has no vertices
of degree 1 or 2.    Most importantly, we claim that $G$ does
\emph{not} contain a strict subgraph $G'$ with $m(G')\geq 2n(G')-1$.
(Call such a subgraph \emph{nearly 4-regular}.)
For if it did, then there must be at least three edges that connect $G'$
to the rest of $G$.    In consequence, there are at least three
vertices of $G'$ with at most three neighbors in $G'$.  By degree-counting,
therefore $2m(G')\leq 4n(G')-3$, which contradicts $m(G')\geq 2n(G')-1$.

To state the next case, we need a definition.   For a vertex-set $S$
of even size, an \emph{anti-matching} is an enumeration of $S$ as $\{s_1,\dots,s_{|S|}\}$
such that for $i=1,\dots,|S|/2$ there is no edge between $s_{2i-1}$ and $s_{2i}$.

\medskip\noindent{\bf Case 2:} There exists a vertex $v$ of degree 4 whose
neighborhood has an anti-matching $\{a,b,c,d\}$.   

Consider the graph $G'=G-v$.   
Observe that no subgraph of $G'$ (including $G'$ itself) is nearly 4-regular,
otherwise this would be a nearly 4-regular strict subgraph of $G$ and Case~1 applies.
Therefore $m(G'')\leq 2n(G'')-2$ for all subgraphs $G''$ of $G'$.\
By Theorem~\ref{thm:NW} graph
$G'$ has arboricity 2 and can be edge-partitioned into two forests
$F_1$ and $F_2$.   Define 
$E_1:=E(F_1)\cup \{(v,a),(v,b)\}$ and $E_2:=E(F_2)\cup \{(v,c),(v,d)\}$.
%where $\{a,b,c,d\}$ is the enumeration of $N(v)$ that is an anti-matching.
(Figure~\ref{fig:octahedronSplit}
illustrates how to apply this on the octahedron to obtain
the edge-partition used in Figure~\ref{fig:octahedron}.) 
Since $F_1$ is a forest that does not reach $v$, adding the two edges incident to $v$ could
close up a cycle, but if it does then there is only one cycle and it contains $a$-$v$-$b$.
Since $(a,b)$ was not an edge, this cycle is not a triangle.  So the graph formed by
$E_1$ is a triangle-free pseudoforest, and symmetrically so is the one formed by $E_2$.

\iffalse
\begin{figure}[ht]
\centering
%\includegraphics[scale=0.1,angle=180,trim=600 100 500 0,clip]{anti_matching.jpg}
\hspace*{\fill}
\includegraphics[scale=0.8,page=1]{antimatching.pdf}
\hspace*{\fill}
\includegraphics[scale=0.8,page=3]{antimatching.pdf}
\hspace*{\fill}
\includegraphics[scale=0.8,page=5]{antimatching.pdf}
\hspace*{\fill}
\caption{The case where $G$ has a degree-4 vertex $v$ with an anti-matching (dashed) in its neighborhood.   
}
\end{figure}
\fi
\begin{figure}[ht]
\centering
\hspace*{\fill}
\includegraphics[scale=0.8,page=2]{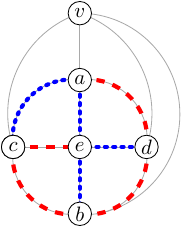}
\hspace*{\fill}
\includegraphics[scale=0.8,page=3]{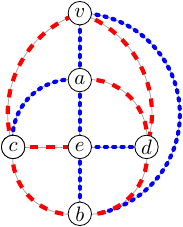}
\hspace*{\fill}
\caption{The case where $G$ has a degree-4 vertex $v$ with an anti-matching $\{a,b,c,d\}$ in its neighborhood.   
}
\label{fig:octahedronSplit}
\end{figure}

\medskip\noindent{\bf Case 3:} There exists a vertex $v$ of degree 3 whose neighborhood
is not a triangle.    This 
case is handled almost exactly as the previous one: we
take two non-adjacent neighbors $a,b$ of $v$, split $G-v$ into two forests,
add $(v,a)$ and $(v,b)$ to one forest and the third
edge of $v$ to the other, and verify all conditions.

\medskip\noindent{\bf Case 4:} None of the above.   We claim that in this case
$G$ is 4-edge-colorable, so has an edge-partition $E^1\cup E^2\cup E^3 \cup E^4$
into four matchings.
This implies the result, because then $E_i:=E^i\cup E^{i+2}$ (for $i\in [2]$) forms
the disjoint union of even-length cycles and paths, and hence is a triangle-free pseudoforest.
Towards 4-edge-colorability, we first need an observation:

\begin{claim} If none of the previous cases applies, then any vertex $v$ belongs to
a complete graph $K_4$.
\end{claim}
\begin{proof}
If $\deg(v)=3$, then its neighbourhood must be a triangle and hence forms a $K_4$ together with $v$.
Now assume that $\deg(v)=4$ and enumerate its neighbors  as $a,b,c,d$, see also Figure~\ref{fig:no_triangle}.
At least one of edges $(a,b)$ and $(c,d)$ must exist (otherwise there is an anti-matching); up to renaming $(a,b)\in E$.
At least one of edge $(a,c)$ and $(b,d)$ must exist (otherwise there is an anti-matching); up to renaming $(a,c)\in E$.
If $(b,c)$ exists then $\{v,a,b,c\}$ forms a $K_4$ and we are done, so assume that $(b,c)\not\in E$.   This implies that
$(a,d)\in E$ (otherwise there is an anti-matching).    If $(b,d)$ exists then $\{v,a,b,d\}$ form a $K_4$ and we are done.   So assume that $(b,d)$ does not exist.

\begin{figure}[ht]
\centering
\includegraphics[scale=1,page=7]{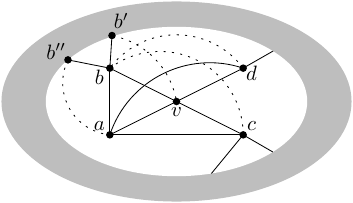}
\caption{The case where a vertex $v$ has no triangle among its neighbors.   Dotted lines indicate a pair of vertices that is not connected by an edge.
}
\label{fig:no_triangle}
\end{figure}

We claim that vertex $b$ can be used for one of the previous cases.   We know that $a,v$ are neighbors of $b$,
and $c,d$ are \emph{not} neighbors of $b$.   If $\deg(b)=2$ then Case~1 applies.   If $\deg(b)\geq 3$, then
let $b'\neq a,v$ be a third neighbor of $b$.   There is no edge $(v,b')$ since $v$ has exactly the
four neighbors $a,b,c,d$ and $b'\neq c,d$ since the latter are not adjacent to $b$.  So if $\deg(b)=3$ then $N(b)$ 
is not a triangle and Case 3 applies.   If $\deg(b)=4$, then let $b''\neq v,a,b'$ be the fourth neighbor
of $b$.   There is no edge $(a,b'')$ since $a$ already has the four neighbors $v,b,c,d$,  and $b''\neq c,d$. So $\{v,b',a,b''\}$
is an anti-matching of $N(b)$ and Case 2 applies.
\end{proof}

\begin{claim}
If none of the previous cases applies, then any two copies $H_1,H_2$ of $K_4$ in $G$ are either identical or vertex-disjoint.
\end{claim}
\begin{proof}
Assume for contradiction that $H_1,H_2$ both contain vertex $v$, but for $i\in [2]$ vertex $w_i$
is in $H_i\setminus H_{3-i}$.   Since $\deg(v)\leq 4$, and it must be adjacent to all other vertices
in $H_1\cup H_2$, we have $|H_1\cup H_2|\leq 5$.
Therefore $H_1$ and $H_2$ share three vertices, and differ only in the
vertices $w_1$ and $w_2$.   It follows that $H_1\cup H_2$ 
%	forms the graph $K_5-e$, i.e., the 
is the complete graph on five vertices with one edge missing.   This is a nearly 4-regular graph, and a strict subgraph of $G$ by $n\geq 6$, and Case~1 applies. 
\end{proof}

Hence the vertices of $G$ can be partitioned into $V_1\cup \dots \cup V_\ell$, where each $V_i$ induces a $K_4$ (and $\ell=n/4$).
The edges of $G$ that are not in the $K_4$'s form a matching, since 
every vertex $v$ has three incident edges inside its $K_4$, and hence at most one edge going outside.
Note that $K_4$ is 3-edge colorable.   So three colors suffice to color the edges inside the $K_4$'s, and
a fourth color suffices to color the matching that connects them.   Hence $G$ is 4-edge-colorable.

\begin{figure}[ht]
\centering
\includegraphics[scale=0.8,page=10]{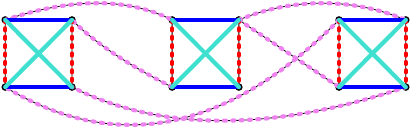}
\caption{The case where $G$ consists of disjoint $K_4$'s. Colors indicate the 4-edge coloring; we also show the resulting edge-partition with dotted/solid edges.
}
\end{figure}

\medskip
As discussed earlier, 4-edge-colorability of $G$ implies that it can be split into two triangle-free pseudoforests, and hence Lemma~\ref{lem:4graph_cactus} holds.

\section{Adding to drawings}
\label{sec:insert_edges}
\label{sec:easy_to_draw}

It remains to show the two other ingredients:   Trees and strict cacti have unbent drawings, and edges can be inserted
into existing drawings if bends are allowed.   Both use related techniques, and both are nearly obvious for orthogonal
drawings as long as one may insert new gridlines.  So we only sketch the proofs here to ensure that  
nothing breaks when building $\d$-grid drawings for arbitrary $\d \geq 2$ instead.

\begin{claim}
\label{claim:addPoint}
Let $\Gamma$ be an $\d$-grid drawing and $v$ be a vertex where one of the gridrays $r$ emanating from $v$ is not
used by an edge.   Then, after scaling the drawing suitably, we can add a new vertex $w$ to $\Gamma$ as well
as an edge $(v,w)$ that is routed along $r$ without bend or crossing.
\end{claim}
\begin{proof}
Walk along ray $r$ from $v$ until we hit another grid point $p$ (this must happen since slopes are rational).
Let $q$ be the first point (if any) during this walk that is occupied by $\Gamma$ (it could be a bend or vertex or
part of an edge-segment).   If there is no such point then we can place $w$ at $p$ and are done, so assume that
$q$ exists.  Define $f$ to be the
smallest integer bigger than $d(v,p)/d(v,q)$, and scale the entire drawing by factor $f$.%
\footnote{For the orthogonal and the hexagonal grid scale-factor $f\leq 2$ is always enough, and for the
octilinear grid $f\leq 3$ suffices.   For arbitrary slope-sets $\calS$ scale-factor $f$ could be quite large,
but it is finite since slopes are rational, and it has an upper bound that only depends on $\calS$ and not
on the size of the graph.}
With this, the
line segment from $v$ to $p$ along $r$ gets divided into $f$ equal-length line segments between gridpoints. 
We can place $w$ at the first of these gridpoints, which gives the result since
$d(v,w)=d(v,p)/f < d(v,q)$ and so no vertex, bend or edge lies on the straight-line segment from $v$ to $w$.
%
%We cannot always place $w$ at $p$, since
%it may already be occupied by another vertex and/or the segment $\overline{vp}$ may be crossed by other edges,
%see also Figure~\ref{fig:addToDrawing}.   Therefore we 
%scale $\Gamma$ by a suitably large factor $f$, i.e., we multiple all coordinates by $f$.   For orthogonal and hexagonal
%drawings $f=2$ suffices, and for octilinear drawings $f=3$ suffices.    For arbitrary $\d$-grids, we choose $f=\lceil 1/\delta\rceil +1$,
%where $\delta$ is the smallest distance between two points where gridlines meet.
%With this, in the scaled drawing
%the ray $r$ emanating from $v$ hits a grid point \emph{before} it has crossed any edge or met a vertex, and so we can
%place $w$ there.
\end{proof}

\begin{figure}[ht]
\centering
\includegraphics[scale=1,page=4]{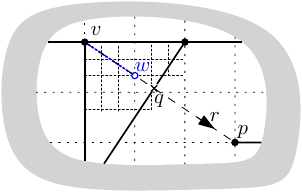}
\caption{Adding a new vertex $w$ after scaling the drawing such that along the specified ray $r$ we
reach a gridpoint before intersecting any element of the drawing.   In this example $f=3$ suffices.}
\label{fig:addToDrawing}
\end{figure}

\begin{claim}
\label{claim:addCycle}
Let $\Gamma$ be an $\d$-grid drawing and $v$ be a vertex with degree less than $2\d$.   Let $C$ be a cycle of vertices
not in $\Gamma$, with $|C|\geq 4$ if $\d=2$.  Then we can add $C$ and an edge $(v,w)$ (for some $w\in C$)
to the drawing without adding a bend or a crossing, after scaling suitably.
\end{claim}
\begin{proof}
The proof is almost exactly the same as for the previous claim, except that we scale the drawing by a bigger factor
such that the ray from $v$ goes through a sufficiently large sub-grid 
%of size $|C|\times |C|$ 
before hitting any other edge
or vertex.   Into this subgrid we can then place cycle $C$, using three different slopes if $|C|=3$. 
\end{proof}

The desired drawing-result (restated here for convenience) is now an easy corollary:

% I don't know why restate doesn't do what it should, but let's hack it so that the numbers look right
\addtocounter{theorem}{-3}
\EasyToDraw*
\addtocounter{theorem}{3}

\begin{proof}
We may assume that the input-graph $G$ is connected, otherwise we can draw each component separately.  We may also assume that it has a degree-1 vertex $r$, otherwise temporarily add one and delete it from the drawing later.
Let $\calT$ be the tree obtained from $G$ by contracting every cycle $C$ into a vertex; 
the result is unique since all cycles are vertex-disjoint.
Root $\calT$ at $r$ and draw $G$ while processing $\calT$ in pre-order.   Drawing $r$ is trivial. To add the subgraph represented by a new child $w$, apply the appropriate one of the two claims above
at its parent $v$.
If $G$ was a plane forest, then $w$ always represents a vertex.   In this case choose a gridray $r$ at its
parent $v$ that is compatible with the rotation scheme, i.e., lies in the correct order between previously used gridrays at $v$, and leaves as many empty gridrays free between $\overline{vw}$ and other used gridrays as are needed for edges that will be placed later.
\end{proof}

The other missing piece (used in the proof of Lemma~\ref{lem:partition_to_drawing}) is how
to add edges to an existing drawing.   This is likewise very easy. 

%\InsertEdge*
%\begin{restatable}{lemma}{InsertEdge}
\begin{lemma}
\label{lem:insert_edges}
Let $\Gamma$ be an $\d$-grid drawing, and let $u,v$ be two vertices with degree less than $2\d$.
Then an edge $(u,v)$ can be added to the drawing.   If furthermore some connected region $R$ of
$\mathbb{R}^2\setminus \Gamma$ contains both $u$ and $v$ and free gridrays into $R$ at each,
then $(u,v)$ can be added without adding a crossing.
\end{lemma}
\begin{proof}
Use Claim~\ref{claim:addPoint} to add a new vertex $b_v$ at a grid-point next to $v$. By forcing the scale-factor $f$ in this proof to be at least 2, at least one gridline at $b_u$ is \emph{new} (added due to scaling) and so contains no other vertex, crossing or bend of $\Gamma$. 
Likewise add a new vertex $b_u$ next to $u$ on a new gridline.     
If the drawing is allowed to have crossings, then now simply add a polygonal curve from $b_u$ to $b_v$ by going view the new gridlines to outside the bounding box of $\Gamma$, and then connect there.%
\footnote{This route was chosen for ease of description; of course much prettier routes
could often be found.} 
If the drawing should be planar and region $R$ exists, then choose the gridrays when placing
$b_v,b_w$ so that the latter two are inside $R$.    After scaling the drawing sufficiently
much, the part of the $\d$-grid that lies inside $R$ becomes connected, and we can hence
connect $b_v$ and $b_w$ while staying inside $R$, hence without crossings (but possibly using lots of bends). 
\end{proof}

\begin{figure}[ht]
\hspace*{\fill}
\includegraphics[page=2,scale=1.1]{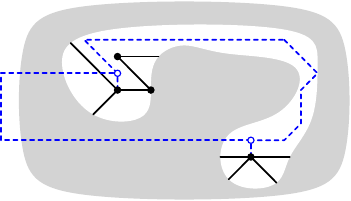}
\hspace*{\fill}
\caption{Inserting a new edge after scaling the drawing suitably: first place bends in adjacent grid points, then connect within a common region (if there is one) or via outside the drawing.}
\label{fig:addEdgeToDrawing}
\end{figure}

%\begin{figure}[ht]
%\hspace*{\fill}
%\includegraphics[page=8,scale=1,trim=0 35 0 0,clip]{octahedron}
%\hspace*{\fill}
%\caption{Adding a missing edge (red dashed) into $\Gamma_i'$.   
%%The edge-route here is determined with the method from Section~\ref{sec:insert_edges}, and far from optimal.
%\todo[inline]{Bend not visible.  And generally probably not bother with complicated routing.}} 
%\end{figure}
%

We can finally complete the proof of Theorem~\ref{thm:hex_oct_planar}, which states
that a planar $\d$-graph $G$ has an unbent collection of $(\d{+}1)$ planar $\d$-grid drawings.
%We finally return to Theorem~\ref{thm:hex_oct_planar} once more, which states
%that a planar $\d$-graph $G$ has an unbent collection of $(\d+1)$ planar $\d$-grid drawings.
%(Its proof in Section~\ref{sec:skeleton} was incomplete since Lemma~\ref{lem:partition_to_drawing}
%cannot be used for planar drawings.)    
Find an edge-partition of $G$ into $d+1$ forests $F_1,\dots,F_{d+1}$ with Lemma~\ref{lem:pseudoforest}.
For $i=1,\dots,d+1$, let $F_i^+$ be the forest obtained by adding $\deg_G(v)-\deg_{F_i}(v)$ many
degree-1 vertices at each vertex $v$, i.e., add `stubs' for each edge of $G$ that is incident to $v$ but not in $F_i$.
Every vertex $v$ now has equally many edges in $G$ and $F_i^+$, and we let it inherit the
planar embedding of $G$ in $F_i^+$.
Obtain a plane drawing of $F_i^+$ with Lemma~\ref{lem:easy_to_draw}.   Since this 
has no crossing and respects the embedding of $G$, any missing edge $e=(v,w)$ can now be inserted
by connecting the two `stubs' that existed for $e$ at $v$ and $w$, using Lemma~\ref{lem:insert_edges}.
This finishes the proof of the theorem.

\section{Outlook}

In this paper, we studied how to create collections of $\d$-grid drawings such that every edge of the graph is drawn without bend in at least one of them.   This can always be done with $\d$ drawings if crossings are allowed, and $\d{+}1$ drawings otherwise.   For $\d=2$ this is best-possible; for $\d\geq 3$ it remains an open problem whether there are graphs that require more than two drawings in such a collection.   ($K_{2\d+1}$ would be a natural candidate to try.)

We note that for all our results we paid no attention to the grid size, especially when scaling in Section~\ref{sec:insert_edges}.  For orthogonal drawings we can always postprocess the result and remove unused gridlines; the non-planar drawings then have size $O(n)\times O(n)$ since we only use a constant number of gridlines for each edge.   For $\d$-grid drawings for $\d\geq 3$ the drawings may well end up with exponential size; reducing this to linear size (or proving it impossible) remains for future study.

A graph-theoretic open question concerns 
Lemma~\ref{lem:4graph_cactus} where we proved that a 4-graph has
an edge-partition into two triangle-free strict cacti.   Can we
replace `strict cacti' by `pseudoforest'?   Also, does the result hold
for higher degrees, i.e., can we split $2\d$-graphs into $\d$ triangle-free
strict cacti for $\d\geq 3$?   The latter would be useful for unbent collection of
orthogonal drawings in $d$-dimensional space (hence for $2d$-graphs). Here 
it is easy to achieve order $d+1$ by partitioning into 
forests, but can we reduce this to order $d$ or even less?

\todo[inline]{Maybe talk about how Lemma~\ref{lem:easy_to_draw} is false for a general cactus?   Use 5-cycle with a 4-cycle attached at every vertex.}

\bibliographystyle{plainurl}
\bibliography{journal,full,gd,unbent,papers}

\end{document}